\documentclass[10pt,conference]{IEEEtran}

\usepackage{eulervm}
\usepackage{graphicx}
\usepackage{array}
\usepackage{amsmath}
\usepackage{amssymb}
\usepackage{amsthm}
\usepackage{mathrsfs}
\usepackage{xfrac}
\usepackage{algorithmic}
\usepackage{algorithm}
\usepackage{upgreek}
\usepackage{multirow}
\usepackage{arydshln}
\usepackage{subfigure}
\usepackage{cite}
\usepackage{clrscode3e}
\usepackage{bm}
\usepackage{epstopdf}
\usepackage{enumerate}

\author{\IEEEauthorblockN{Lei You$^1$, Lei Lei$^1$, and Di Yuan$^{1,2}$}
\IEEEauthorblockA{\small$^1$Department of Science and Technology,
Link\"{o}ping University, Sweden\\}
\IEEEauthorblockA{\small$^2$Institute for Systems Research, University of Maryland, College Park, MD, 20740, USA\\}
\small\texttt{\{lei.you, lei.lei, di.yuan\}@liu.se, diyuan@umd.edu}
}

\begin{document}

\title{Load Balancing via Joint Transmission in Heterogeneous LTE: Modeling and Computation}
\maketitle

\begin{abstract}
As one of the Coordinated Multipoint (CoMP) techniques, Joint Transmission (JT) can improve the overall system performance. 
In this paper, from the load balancing perspective, we study how the maximum load can be reduced by optimizing JT pattern that characterizes the association between cells and User Equipments (UEs). To give a model of the interference caused by cells with different time-frequency resource usage, we extend a load coupling model, by taking into account JT. In this model, the mutual interference depends on the load of cells coupled in a non-linear system with each other. Under this model, we study a two-cell case and proved that the optimality is achieved in linear time in the number of UEs. After showing the complexity of load balancing in the general network scenario, an iterative algorithm for minimizing the maximum load, named JT-MinMax, is proposed. We evaluate JT-MinMax in a Heterogeneous Network (HetNet), though it is not limited to this type of scenarios. Numerical results demonstrate the significant performance improvement of JT-MinMax on min-max cell load, compared to the conventional non-JT solution where each UE is served by the cell with best received transmit signal.
\end{abstract}

\maketitle

\theoremstyle{plain}
\newtheorem{definition}{Definition}
\newtheorem{theorem}{Theorem}
\newtheorem{lemma}{Lemma}
\newtheorem{proposition}{Proposition}
\newtheorem{postulation}{Postulation}
\newtheorem{property}{Property}
\newtheorem{observation}{Observation}
\newtheorem{corollary}{Corollary}

\section{Introduction}

In Long Term Evolution (LTE), Heterogeneous Networks (HetNets) are viewed as an attractive approach for expanding mobile network capacity as well as alleviating the traffic burden \cite{Andrews:ez}. In a HetNet, both overlaying Macro Cells (MCs) and underlying Small Cells (SCs) are deployed. MCs provide wide area data services and SCs offload part of the traffic volume from MCs. For the forthcoming 5G, the concept of ultra dense HetNet is proposed \cite{Andrews:ez}. As one of the Coordinated Multipoint (CoMP) techniques, Joint Transmission (JT) is viewed as an important technique to incorporate large amount of SCs into 5G networks \cite{Porcello:2014wq, Zhang:2012ej}. In JT, multiple geographically separated cells are allowed to transmit data to a User Equipment (UE) simultaneously \cite{Tolli:2009gg}, so as to enhance the Signal-to-Interference-and-Noise-Ratio (SINR). On the other hand, the resource blocks for the data transmission to this UE are consumed in multiple cells. From this point, the strategies for JT should be chosen carefully with respect to the resource block utilization.

Cell load balancing is a key role in radio resource optimization in HetNet. A model characterizing the cell load is introduced in \cite{Siomina:2009bp}, where the load is defined as the proportion of resource blocks consumption on each cell. Rather than taking the interfering cells as either fully loaded or fixed in a constant load level, this model characterizes the property that interference caused by one cell depends on the average resource utilization of this cell. The cell load indicates the likelihood of one cell that receives the interference from another cell. In \cite{Fehske:2012iw}, it has been shown an accurate interference dependent model for the network performance evaluation. This paper extends the model by taking into consideration JT. Under the extended model, the transmission pattern in JT characterizing the association between cells and UEs affects the network performance in a complicated way. 
For example, letting a cell expand to serve a UE via JT can offload part of the traffic burden from other cells. However, the interference caused by this cell to other cells is increased, because more resource blocks are consumed on this cell. Then the overall performance might go worse. 
 
There are a few recent investigations for cell load balancing problems and JT in HetNets. In \cite{Ye:2013ha}, a user association scheme for load balancing in HetNets is proposed. In \cite{Siomina:2012bl}, by SC range assignment, a load balancing algorithm is proposed in consideration of the load coupling. In \cite{Tanbourgi:2014ck}, JT cooperation is modeled and analyzed for HetNet. However, as far as we know, few work has been done on optimizing JT pattern in such a load coupling model. Further, the load coupling model proposed in \cite{Siomina:2009bp} does not apply directly to the JT scenarios. In this paper, we extends this model, and investigate how to optimize the JT pattern so as to minimize the maximum cell load. The contributions are summarized as follows.
\vskip -20pt
\begin{enumerate}
	\item \emph{A generalized load coupling model}. The proposed model differs from the previous work in the following aspects. First, we extend the previous load coupling model by taking into consideration the JT scenarios. This model characterizes the influence caused by the change of JT pattern on load. 
	\item \emph{Theoretical analysis for a two-cell case}. For this case, two cells are deployed in the network and the two cells serve the same number of UEs. The transmit power and the channel gain of the two cells are symmetric. It is proved that the global optimality is achieved in linear time in the number of UEs.
	\item \emph{Load balancing for the general scenarios}. We showed the computational complexity of load balancing in the general case. A sufficient condition for min-max cell load is given. Based on the condition, an heuristic algorithm JT-MinMax is proposed.
\end{enumerate}
%
 

\section{System Model and Generalized Load Coupling}
\label{sec:sys_mod}

\subsection{Basic Notations}
Denote the sets of all cells by $\mathcal{I}$. Denote by $\mathcal{I}'$ and $\mathcal{I}\backslash\mathcal{I}'$ the sets of all MCs and SCs, respectively. Denote the set of UEs by $\mathcal{J}$. Let $n=|\mathcal{I}|$ and $m=|\mathcal{J}|$. Each UE can be served simultaneously by more than one cells. Let $\mathcal{I}_j$ denote the set of cells serving UE $j$, and $\mathcal{J}_i$ the set of UEs served by cell $i$, respectively. We exclude the case that the UE's demand is zero, so that each UE is served by at least one cell, i.e., $|\mathcal{I}_j|\geq 1,~\forall j\in\mathcal{J}$. The JT pattern is given by an $n\times m$ matrix $\bm{\kappa}$, where $\kappa_{ij}=1$ means that cell $i$ is currently serving UE $j$. We have $|\mathcal{I}_j|=\sum_{i=1}^{n}\kappa_{ij}$ and $|\mathcal{J}_i|=\sum_{j=1}^{m}\kappa_{ij}$, for $\forall i\in\mathcal{I}$ and $j\in\mathcal{J}$, respectively.

\subsection{Load Coupling with JT}
For convenience, we let the JT pattern be fixed in the expression of the SINR and the cell load, in this subsection. 

\subsubsection{SINR in JT}

We model the SINR $\gamma_j$ of UE $j$ in Eq.~(\ref{eq:SINR}). 
\begin{equation}
\gamma_j=\frac{\sum_{i\in \mathcal{I}_j}p_ig_{ij}}{\sum_{k\in \mathcal{I} \backslash \mathcal{I}_j}p_kg_{kj}x_k+\sigma^2}
\label{eq:SINR}
\end{equation}

In Eq.~(\ref{eq:SINR}), the transmit power per resource block (in time and frequency) of cell $i$ is $p_i$ ($p_i>0$), and $g_{ij}$ is the channel gain between cell $i$ and UE $j$. And $\sum_{i\in\mathcal{I}_j}p_ig_{ij}$ is the received signal power from all the serving cells $\mathcal{I}_j$ of the UE $j$. In the denominator, $\sigma^2$ is the noise power. Entity $x_k$ is the load of cell $k$, which is defined to be the proportion of resource blocks consumed on cell $k$ by all the UE $j\in\mathcal{J}_k$. In this context, $x_k$ is intuitively interpreted as the likelihood that the served UEs of cell $i$ receive the interference from $k$ on all resource blocks. Thereby, the term $\sum_{k\in \mathcal{I} \backslash \mathcal{I}_j}p_kg_{kj}x_k$ is the interference that UE $j$ received from other cells.

\subsubsection{Cell's Load for a given Bitrate Demand}

The load of any cell $i$, is represented in Eq.~(\ref{eq:load}), in concern of the bitrate demand of its served UEs.
\begin{equation}
x_i=\sum_{j\in\mathcal{J}_i}y_j,~\textnormal{where }y_j=\frac{d_j}{MB\log_2\left(1+\gamma_j\right)}
\label{eq:load}
\end{equation}

In Eq.~(\ref{eq:load}), $d_j$, is the bitrate demand of UE $j$. In the denominator, $B$ is the bandwidth per resource block and $M$ is the total number of resource blocks available in each cell. The entity $B\log_2(1+\gamma_j)$ is the achievable bitrate per resource block and thus $MB\log_2(1+\gamma_j)$ is the total achievable bitrate for UE $j$. Then $y_j$ is the proportion of the resource block consumption of UE $j$ on all its serving cells. Note that the resource blocks of any UE $j$ consumed on each of its serving cells $i\in\mathcal{I}_j$ are equal, which conforms to that the transmitted data from all cells to a certain UE should be the same in JT. This can be verified in Eq.~(\ref{eq:load}). If UE $j$ is served by cell $i$ and cell $k$ simultaneously, then the consumed resource blocks $y_j$ appears as a term in both $x_i$ and $x_k$. In the remaining part of this paper, for simplicity, we let $MB = 1$ without loss of generality, and $d_{j}$ is normalized by $MB$. We remark that this is an approximate interference coupling model, in which all of the cell load is counted as the likelihood that other cells receive the interference from this cell. Actually, the UEs served by a certain cell at the same time do not share the resource blocks thus generating no mutual interference among each other. For simplicity, we treat all the cell load as the interference part in this model.



\subsection{Load Coupling Model Characterizing JT pattern}

In this subsection, we extended the load coupling model to characterize the JT pattern. We define the \emph{cell load function} and the \emph{SINR function} in  Eq.~(\ref{eq:cell_load_function}) and  Eq.~(\ref{eq:UE_load_function}), respectively.
\begin{equation}
{f}_i^{\bm{\kappa}}(\bm{\gamma}):=\sum_{j=1}^{m}\frac{\kappa_{ij}d_j}{\log_2\left(1+\gamma_j\right)}
\label{eq:cell_load_function}
\end{equation}
\begin{equation}
{h}_j^{\bm{\kappa}}(\bm{x}):=\frac{\sum\limits_{i=1}^{n}p_ig_{ij}\kappa_{ij}}{\sum\limits_{k=1}^{n}p_kg_{kj}x_{k}(1-\kappa_{kj})+\sigma^2}
\label{eq:UE_load_function}
\end{equation}
For $f_i^{\bm{\kappa}}$ and $h_i^{\bm{\kappa}}$, the set $\mathcal{J}_i$ in  Eq.~(\ref{eq:load}) and $\mathcal{I}_j$ in  Eq.~(\ref{eq:SINR}) are indicated by $\bm{\kappa}$.  We let $\bm{f}^{\bm{\kappa}}(\cdot)=[f_1^{\bm{\kappa}}(\cdot),f_2^{\bm{\kappa}}(\cdot),\ldots,f_n^{\bm{\kappa}}(\cdot)]$ and $\bm{h}^{\bm{\kappa}}(\cdot)=[h_1^{\bm{\kappa}}(\cdot),h_2^{\bm{\kappa}}(\cdot),\ldots,h_n^{\bm{\kappa}}(\cdot)]$. Thereby we have the following load coupling equation, shown in  Eq.~(\ref{eq:coupling_1}).
\begin{equation}
\bm{x}=\bm{f}^{\bm{\kappa}}\circ\bm{h}^{\bm{\kappa}}(\bm{x})
\label{eq:coupling_1}
\end{equation}

In Eq.~(\ref{eq:coupling_1}), symbol ``$\circ$'' represents the compound relationship between functions. In other words, $\bm{f}\circ \bm{h}(\bm{x})$ means $\bm{f}[\bm{h}(\bm{x})]$. Some of our results essentially rely on the framework of the standard interference functions (SIF) introduced by Yates in \cite{Yates:1995eh} and has been further studied by Schubert et. al. in \cite{Schubert:2014ud}. After giving definition of SIF, we show some properties and an observation.

\begin{definition}
A function $\bm{f}$: $\mathbb{R}^m_+\rightarrow\mathbb{R}_{++}$ is called an SIF if the following properties hold:
\begin{enumerate}
\item (Scalability) $\alpha \bm{f}(\bm{x})>\bm{f}(\alpha\bm{x}),~\forall \bm{x}\in \mathbb{R}^m_+,~\alpha>1$.
\item (Monotonicity) $\bm{f}(\bm{x})\geq \bm{f}(\bm{x'})$, if $\bm{x} \geq \bm{x'}$.
\end{enumerate}
\end{definition}
\begin{property}
A standard interference function $\bm{f}$ has the following properties \textnormal{\cite{Schubert:2014ud}}:
\begin{enumerate}
\item The function $\bm{f}$ has a fixed point $\bm{x^*}$ if and only if there exists $\bm{x'}\in\mathbb{R}_+^m$ satisfying $\bm{f}(\bm{x'})\leq\bm{x'}$.
\item  For the sequence $\bm{x}^{(0)},\bm{x}^{(1)},\ldots$ generated by fix-point iteration, if there exists $k$ satisfying $\bm{f}(\bm{x}^{(k)}) \leq\bm{f}(\bm{x}^{(k+1)})$, then the sequence $\bm{x}^{(k)},\bm{x}^{(k+1)},\ldots$ is monotonously decreasing (in every component).
\end{enumerate}
\end{property}

\begin{observation}
Given $\bm{\kappa}$, ${f}^{\bm{\kappa}}_i\circ\bm{h}^{\bm{\kappa}}(\bm{\cdot})$ is an SIF of $\bm{x}$ if 1) there exist $i\in[1,n]$, $j\in[1,m]$ such that $\kappa_{ij}=1$ and 2) $\sum_{k=1}^{n}\kappa_{kj}<n$.
\label{ob:varphi_SIF}
\end{observation}
For any cell $i$, if there exists $i\in[1,n]$ and $j\in[1,m]$ such that $\kappa_{ij}=1$, then $f_i(\cdot)$ is a function of $h_j$. For any UE $j$, if $\sum_{k=1}^{n}\kappa_{kj}<n$, then there is at least one cell $k$ such that $\kappa_{kj}=0$. For this cell $k$, there exist a UE $j'\neq j$ such that $\kappa_{ij'}=1$. Therefore, ${h}_j(\cdot)$ is a function of $\bm{x}$. We can verify that the function $f_i^{\bm{\kappa}}\circ\bm{h}^{\bm{\kappa}}(\bm{x})$ is concave for $\bm{x}$. (Due to the space the proof is not given here but will be published elsewhere). By the conclusion in~\cite{Porcello:2014wq} that any concave function is an SIF, we get that $\bm{f}^{\bm{\kappa}}\circ\bm{h}^{\bm{\kappa}}(\bm{x})$ is an SIF of $x$. It is shown in \cite{Yates:1995eh} that an SIF Eq.can be solved by fixed-point iterations. Thus for each given JT pattern, we can compute the corresponding network-wide cell load by doing fixed point iterations in $\bm{f}^{\bm{\kappa}}\circ\bm{h}^{\bm{\kappa}}(\bm{x})$.

\vskip -10pt
\section{Load Balancing for a two-cell Case}
\label{sec:special}

We first investigate a symmetric two-cell case that is tractable, as an introduction to the load balancing problem in the load coupling model. There are cell $1$ and cell $2$ in the network. We denote the set of the two cells' served UE by $1,2,\ldots,m$ and $m+1,m+2,\ldots,2m$ respectively. The symmetry is reflected in the following aspects. First, the transmit power of cell $1$ and $2$ are the same, i.e., $p_1=p_2$. Second, for every UE $j\in[1,m]$, we have the channel gain satisfying $g_{1,j}=g_{2,m+j}$ and $g_{1,m+j}=g_{2,j}$. Besides, we have $d_{j}=d_{m+j}$ as the user demand of any UE $j\in[1,m]$. 

\subsection{Formulation}
The load of any UE $j$ served simultaneously by both cell 1 and 2 is a constant, shown in Eq.~(\ref{eq:constant_load}). This is because $\gamma_j$ is independent of both $x_1$ and $x_2$ in this case.
\begin{equation}
c_j:=\frac{d_j}{\log_2\left(1+\frac{p_1g_{1j}+p_2g_{2j}}{\sigma^2}\right)}
\label{eq:constant_load}
\end{equation}
We use the vector $\bm{\kappa}=[\kappa_1,\kappa_2,\ldots,\kappa_m]$ to denote the JT pattern, i.e., $\kappa_j=1$ means UE $j$ is served by both cell $1$ and $2$. Otherwise, UE $j$ is served only by cell 1 or cell 2, depending on whether $j$ is larger than $m$. Let $y_j,j\in[1,m]$ be the load that UE $j$ consumed on cell $1$. Let $y_j,\j\in[m+1,2m]$ be the load that UE $j$ consumed on cell $2$. The load coupling Eq.is re-written as  Eq.~(\ref{eq:two_cells_coupling1}) and  Eq.~(\ref{eq:two_cells_coupling2}).
\begin{equation}
x_1(\bm{\kappa},x_2):=\!\!\sum_{j=1}^my_j,~y_j=\frac{(1-\kappa_j)d_j}{\log_2\left(1+\frac{p_1g_{1j}}{p_2g_{2j}x_2+\sigma^2}\right)}+c_j\kappa_j
\label{eq:two_cells_coupling1}
\end{equation}
\begin{equation}
\!\!x_2(\bm{\kappa},x_1):=\!\!\!\!\!\!\sum_{j=m+1}^{2m}\!\!\!y_j,~y_j=\frac{(1-\kappa_j)d_j}{\log_2\left(1+\frac{p_2g_{2j}}{p_1g_{1j}x_2+\sigma^2}\right)}+c_j\kappa_j
\label{eq:two_cells_coupling2}
\end{equation}
Note that $x_1$ and $x_2$ are mutually coupled with each other. In addition, both $x_1(\cdot)$ and $x_2(\cdot)$ are SIF in $x_2$ and $x_1$, respectively. That means, if the load of either cell is reduced, then that of the other will be also reduced. The load balancing problem is formalized in Eq.~(\ref{eq:p1}).
\vskip -10pt
\begin{subequations}
\begin{alignat}{2}
 [\textbf{MinMaxL-S}]~\min\limits_{x_1,x_2,\bm{\kappa}} &\quad \eta & \\
 \textnormal{s.t.}  &\quad  x_1=x_1(\bm{\kappa},x_2) &   \\
 &\quad  x_2=x_2(\bm{\kappa},x_1)&   \\
 &\quad x_1,x_2\leq \eta  &\\
  & \quad 0<x_{1},x_{2}\leq 1 &   \\
 & \quad \kappa_{j} \in \{0,1\} &   \quad \!\!\!\!\forall j\in[1,2m] 
\end{alignat}
\label{eq:p1}
\end{subequations}
In MinMaxL-S, $\eta$ is the maximum cell load that we want to minimize. Load coupling constraints are shown in (\ref{eq:p1}b) and (\ref{eq:p1}c).


\subsection{Main Results}


\begin{lemma}
In MinMaxL-S, $y_{j}=y_{m+j}$ at convergence if $\kappa_{j}=\kappa_{m+j}$ for all $j\in[1,m]$.
\label{lma:equal_load}
\end{lemma}
\begin{proof}
If $\kappa_j=\kappa_{m+j}=1$, then $y_{j}=y_{m+j}=c_j$. Now we focus on the case of $\kappa_j=\kappa_{m+j}=0$. Suppose $y_{j}>y_{m+j}$. (The proof for the less-than case is similar.) Then we have $x_2>x_1$, since 
\begin{enumerate}
	\item $y_j$ and $y_{m+j}$ is monotonically increasing function in $x_2$ and $x_1$, respectively
	\item both $y_j$ and $y_{m+j}$ are symmetric in all other parameters with each other
\end{enumerate}
Therefore, for any $k\neq j$ with $\kappa_k=\kappa_{m+k}=1$, we have $y_{k}>y_{m+k}$, which leads to $x_1>x_2$, conflicting the former results.
\end{proof}

\begin{definition}
The \textbf{symmetric JT rule} is that, when cell $2$ expands to serve any UE $j\in[1,m]$, cell $1$ also expands to serve the UE $m+j$.
\label{def:symmetric_CoMP}
\end{definition}

%

\begin{lemma}
For any non-symmetric JT pattern, there exists a corresponding symmetric JT pattern achieves the lower maximum load.\label{lma:optimal_CoMP}	
\end{lemma}
\begin{proof}
Let $\bm{\kappa}=[\bm{\kappa^1},\bm{\kappa^2}]$ be a non-symmetric JT pattern, where $\bm{\kappa^1}=[\kappa_1,\kappa_2,\ldots,\kappa_m]$ and $\bm{\kappa^2}=[\kappa_{m+1},\kappa_{m+2},\ldots,\kappa_{2m}]$ with $\bm{\kappa^1}\neq \bm{\kappa^2}$. Denote by $\overline{x}_1$ and $\overline{x}_2$ the load at convergence, with $\bm{\overline{\kappa}}=\bm{0}$. Denote by $x_1$ and $x_2$ the load before any iteration after we change $\bm{\overline{\kappa}}$ to $\bm{\kappa}$. That is, $x_1=x_1(\bm{\kappa}^1,\overline{x}_2)$ and $x_2=x_2(\bm{\kappa}^2,\overline{x}_1)$.  Suppose $\overline{x}_1>\overline{x}_2$ (The proof for the less-than case is similar).
Let $\bm{\kappa^{1'}}=\bm{\kappa^2}$, and then $\bm{\kappa'}=[\bm{\kappa^{1'}},\bm{\kappa^2}]$ is a symmetric JT pattern. Let $x'_1=x_1(\bm{\kappa^{1'}},\overline{x}_2)$ and $x'_2=x_2(\bm{\kappa^{2}},\overline{x}_1)$. Due to the symmetry, we have $x'_1=x_1$ and $x'_2=x'_1=x_1$. 

For any $y_j$ in $x_1$ and $y_k$ in $x_2$, we have $y_j>y_k$, because of $x_2>x_1$. Therefore, in the next iteration of the non-symmetric case with $\bm{\kappa}$, $x_1$ will increase, which further causes $x_2$ increase. Let $x_1^{*}$ and $x_2^{*}$ be the load after we change $\bm{\overline{\kappa}}$ to $\bm{\kappa}$ at convergence. We have $x_1^{*}>x'_1$. For the symmetric JT case with $\bm{\kappa'}$, according to Lemma~\ref{lma:equal_load} and the unique fix-point property of SIF, $x'_1=x'_2$ is at convergence, which is less than both $x_1^{*}$ and $x_2^{*}$. Thus the conclusion.
\end{proof}

\begin{definition}
The \textbf{gain of load} for any UE $j\in[1,2n]$ is defined as
$
G_j:=\overline{y}_{j}-2c_j
$
where $\overline{y}_j$ is the load with $\bm{\kappa}=\bm{0}$.
\end{definition}

\begin{theorem}
(\textbf{Greedy Selection)} Suppose $\eta$ and $\eta'$ are the maximum load at convergence for $\bm{\kappa}$ and $\bm{\kappa'}$, respectively. Under the symmetric JT rule, $\eta'<\eta$ if and only if $\sum_{j=1}^{j=2m}\kappa'_jG_j>\sum_{j=1}^{j=2m}\kappa_jG_j$. 
\label{thm:nec_suf}
\end{theorem}
\begin{proof}
By Lemma \ref{lma:equal_load} and Lemma \ref{lma:optimal_CoMP}, we have 
$
\sum_{j=1}^{j=m}\kappa_jG_j=\sum_{j=m+1}^{j=2m}\kappa_jG_j
$
and 
$
\sum_{j=1}^{j=m}\kappa'_jG_j=\sum_{j=m+1}^{j=2m}\kappa'_jG_j
$ 
Then we focus on $\sum_{j=1}^{j=m}\kappa_jG_j$ and $\sum_{j=1}^{j=m}\kappa'_jG_j$. For the necessity, we prove its converse-negative proposition. If $\sum_{j=1}^{j=2m}\kappa'_jG_j<\sum_{j=1}^{j=2m}\kappa_jG_j$, then the new convergence points for $\bm{\kappa}$ and $\bm{\kappa'}$ are 
$
x_1=\overline{x}_1-\sum_{j=1}^{j=2m}\kappa_jG_j
$ and 
$
x'_1=\overline{x}-\sum_{j=1}^{j=2m}\kappa'_jG_j
$
respectively. And we have $x'_1>x_1$ thus $\eta'>\eta$. The proof of the sufficiency is the same as that for the necessity.
\end{proof}



By Theorem \ref{thm:nec_suf}, it is shown that a greedy algorithm achieves the global optimality for MinMaxL-S problem, by utilizing the greedy selection rule on $G_j$. That is, to let $\kappa_j=1$ if $G_j>0$, for all $j\in[1,2m]$.

\section{A General Load Balancing Algorithm}
\label{sec:general}


\subsection{Formulation}

The formulation for load balancing in the general case is shown in Eq.~(\ref{eq:p0}). The objective is to minimize the maximum cell load $\eta$. The optimization variable $\bm{\kappa}$ gives the JT pattern for each UE. The load coupling Eq.is in constraint (\ref{eq:p0}b). Constraint (\ref{eq:p0}c) gives the definition of the maximum load $\eta$. In constraint (\ref{eq:p0}d), the load of each cell is limited to $1$ at most. In constraint (\ref{eq:p0}e), each UE is limited to be served by at most $K(K<n)$ cells at a time. Constraint (\ref{eq:p0}f) appoints the domain of definition for $\bm{\kappa}$. 
\vskip -10pt
\begin{subequations}
\begin{alignat}{2}
 [\textbf{MinMaxL-G}]~~~ \min\limits_{\bm{x,\kappa}} &\quad \eta \\
 \textnormal{s.t.}  &\quad  \bm{x}=\bm{f}^{\bm{\kappa}}\circ\bm{g}^{\bm{\kappa}}(\bm{x})\\
 &\quad x_{i}\leq \eta ~~\quad\quad\quad\forall i\in\mathcal{I} \\
 &\quad 0< x_{i}\leq 1 \quad\quad\forall i\in\mathcal{I} \\
 &\quad \sum_{i=1}^{n}\kappa_{ij}\leq K \quad~\forall j\in\mathcal{J} \\
 &\quad \kappa_{ij}\in\{0,1\} \quad\forall i\in\mathcal{I}, j\in\mathcal{J}
\end{alignat}
\label{eq:p0}
\end{subequations}

\vskip -12pt

\subsection{Main Results}

%
The first main result is the computational complexity of MinMaxL-G shown in Theorem \ref{thm:hardness}.

\begin{theorem}
MinMaxL-G is $\mathcal{NP}$-hard.
\label{thm:hardness}
\end{theorem}
\begin{proof}
The complexity of MinMaxL-G is proved by a reduction from the 3-SAT problem. Due to the limit of the paper length, the proof detail is not shown here, but will be published in a journal version.
\end{proof}



\begin{lemma}
Suppose $\kappa_{ij}=0\longrightarrow\kappa'_{ij}=1$. Then $\forall\bm{x},~\bm{{f}}\circ\bm{{h}^{\bm{\kappa'}}}(\bm{x})\leq\min\left\{\bm{{f}}\circ\bm{{h}}(\bm{x}),\bm{{f}}^{\bm{\kappa'}}\circ\bm{{h}}^{\bm{\kappa'}}(\bm{x})\right\}$.	
\label{lma:x_bounds}
\end{lemma}
\begin{proof}
On one hand, $\bm{{h}}^{\bm{\kappa'}}(\bm{x})\geq\bm{{h}}(\bm{x})\Longrightarrow\bm{{f}}\circ\bm{{h}}^{\bm{\kappa'}}(\bm{x})\leq\bm{{f}}\circ\bm{{h}}(\bm{x})$. On the other hand, $\bm{{f}}^{\bm{\kappa'}}(\bm{\gamma})\geq\bm{{f}}(\bm{\gamma})\Longrightarrow\bm{{f}}^{\bm{\kappa'}}\circ\bm{{h}}^{\bm{\kappa'}}(\bm{x})$. 
\end{proof}

We show a sufficient condition in Theorem \ref{thm:adding_sufficient}, for improving the maximum cell load $\eta$ in MinMaxL-G, corresponding to adding a JT downlink between cell $i$ and UE $j$. Suppose we change the element $\kappa_{ij}$ in $\bm{\kappa}$ from 0 to 1 and denote the obtained pattern by $\bm{\kappa'}$. This operation is denoted by $\kappa_{ij}=0\longrightarrow\kappa'_{ij}=1$.

\begin{theorem}
Suppose $\kappa_{ij}=0\longrightarrow\kappa'_{ij}=1$, $\bm{\widetilde{x}}=\bm{{f}}\circ\bm{{h}}(\bm{\widetilde{x}})$ and $\bm{x}=\bm{{f}^{\bm{\kappa'}}}\circ\bm{{h}^{\bm{\kappa'}}}(\bm{x})$. Then $\bm{x}\leq\bm{\widetilde{x}}$ if $\exists k\geq 1$ in the iteration $\bm{x}^{(k)}=\bm{{f}}\circ\bm{{h}}^{\bm{\kappa'}}(\bm{x}^{(k-1)})$ such that ${f}^{\bm{\kappa'}}_c\circ\bm{{h}}^{\bm{\kappa'}}(\bm{x}^{(k)})\leq x^{(k)}_c$, where $\bm{x}^{(0)}=\bm{\widetilde{x}}$.
\label{thm:adding_sufficient}
\end{theorem}
\begin{proof} The basic idea is to construct an iteration process that making the cell load vector in iterations monotonically decreases. In iterations $t\in[1,k]$, let $\bm{x}^{(t)}=\bm{{f}}\circ\bm{{h}}^{\bm{\kappa'}}(\bm{x}^{(t-1)})$. By Lemma~\ref{lma:x_bounds}, we have
$
\bm{x}^{(1)}=\bm{{f}}\circ\bm{{h}}^{\bm{\kappa'}}(\bm{x}^{(0)})\leq\bm{{f}}\circ\bm{{h}}(\bm{x}^{(0)})=\bm{x}^{(0)}
$
By Property 1, part 2), we have
\begin{equation}
\bm{x}^{(k)}\leq\bm{x}^{(k-1)}\leq\cdots\leq\bm{x}^{(0)}
\label{eq:k_leq_k-1}
\end{equation}
In iterations $t>k$, let $\bm{x}^{(k+1)}=\bm{{f}}^{\bm{\kappa'}}\circ\bm{{h}}^{\bm{\kappa'}}(\bm{x}^{(k)})$. According to the condition in Theorem~\ref{thm:adding_sufficient}, ${f}^{\bm{\kappa'}}_c\circ\bm{{h}}^{\bm{\kappa'}}(\bm{x}^{(k)})\leq x^{(k)}_c$ holds. For any $i\neq c$, by  Eq.~(\ref{eq:cell_load_function}) and  Eq.~\ref{eq:UE_load_function}, we have 
${f}^{\bm{\kappa'}}_i\circ\bm{{h}}^{\bm{\kappa'}}(\bm{x}^{(k)})={f}_i\circ\bm{{h}}^{\bm{\kappa'}}(\bm{x}^{(k)})$. 
According to the monotonicity and  Eq.~(\ref{eq:k_leq_k-1}),
$
{f}_i\circ\bm{{h}}^{\bm{\kappa'}}(\bm{x}^{(k)})\leq{f}_i\circ\bm{{h}}^{\bm{\kappa'}}(\bm{x}^{(k-1)})=\bm{x}^{(k)}
$
holds. Then we have 
$
\bm{x}^{(k+1)}=\bm{{f}}^{\bm{\kappa'}}\circ\bm{{h}}^{\bm{\kappa'}}(\bm{x}^{(k)})\leq\bm{x}^{(k)}
$
By Property 1, part 1) and 2), at convergence we have
$
\bm{x}=\bm{{f}}^{\bm{\kappa'}}\circ\bm{{h}}^{\bm{\kappa'}}(\bm{x})\leq\cdots\leq\bm{x}^{(k+1)}\leq\bm{x}^{(k)}
\label{eq:x_converge}	
$.
Combined with  Eq.~(\ref{eq:k_leq_k-1}), we have $\bm{x}\leq\bm{x}^{(0)}=\bm{\widetilde{x}}$. 
Hence the conclusion.
\end{proof}

\subsection{Algorithm Design}

As shown in Algorithm \ref{alg:general}, the basic idea of JT-MinMax is to check if a better maximum cell load can be obtained by the sufficient condition given in Theorem~\ref{thm:adding_sufficient}. Parameter $\gamma$ is an integer, used as the counting variable. JT-MinMax goes over the network for $\gamma$ iterations. 
Parameter $\tau$ is the pre-assigned maximum number of iterations on checking the sufficient condition. Both $\gamma$ and $\tau$ affect the performance of JT-MinMax. Basically, the larger $\gamma$ and $\tau$ guarantee the better solution, while on the other hand increasing the computation effort.

\begin{algorithm}[h]
\renewcommand{\algorithmicrequire}{\textbf{Given: }}
\renewcommand{\algorithmicensure}{\textbf{Output: }}
\begin{algorithmic} 
\caption{JT-MinMax}
\label{alg:general}
\REQUIRE $\bm{p}$, $\bm{d}$, $\bm{w}$, $\bm{x}$, $\bm{y}$, $\tau$, $\gamma$\\
\ENSURE $\bm{\kappa^{*}}$, $\bm{x^{*}}$\\
\begin{codebox}
\li \While $\gamma>0$  \Do
\li 	\For $i \leftarrow 1 \To n$ and $j \leftarrow 1 \To n$ \Do
\li 			\If $\kappa_{ij}= 0$ \Then
\li                 $\bm{\kappa'}\leftarrow\bm{\kappa}:\kappa_{ij}\leftarrow1$
\li					\For $k \leftarrow 1 \To \tau$ \Do
\li						\If $\sum\limits_{h=1}^n\kappa_{hj}\leq K$ \Then
\li							$\bm{x}^{(k)}\leftarrow\bm{{f}}\circ\bm{{h}}^{\bm{\kappa}}(\bm{x}^{(k-1)})$
\li							\If ${f}^{\bm{\kappa'}}_c\circ\bm{{h}}^{\bm{\kappa'}}(\bm{x}^{(k)})\leq x^{(k)}_c$  \Then
\li								$\bm{\kappa}\leftarrow\bm{\kappa'}$		
\li								\textbf{break}
						\End
					\End
				\End
			\End
		\End
\li  	$\gamma\leftarrow \gamma-1$
    \End
\li	$\bm{\kappa}^{*}\leftarrow\bm{\kappa}$
\li $\bm{x^{*}}=\bm{f}^{\bm{\kappa^*}}\circ\bm{g}^{\bm{\kappa^*}}(\bm{x^{*}})$
\li \Return $\bm{\kappa'}$, $\bm{x^{*}}$    
\end{codebox}
\end{algorithmic} 
\end{algorithm}

\section{Simulation}
\label{sec:simulation}

\begin{figure*}
\vskip -16pt
\centering
\subfigure[A HetNet layout example.]
{\includegraphics[width=0.25\linewidth]{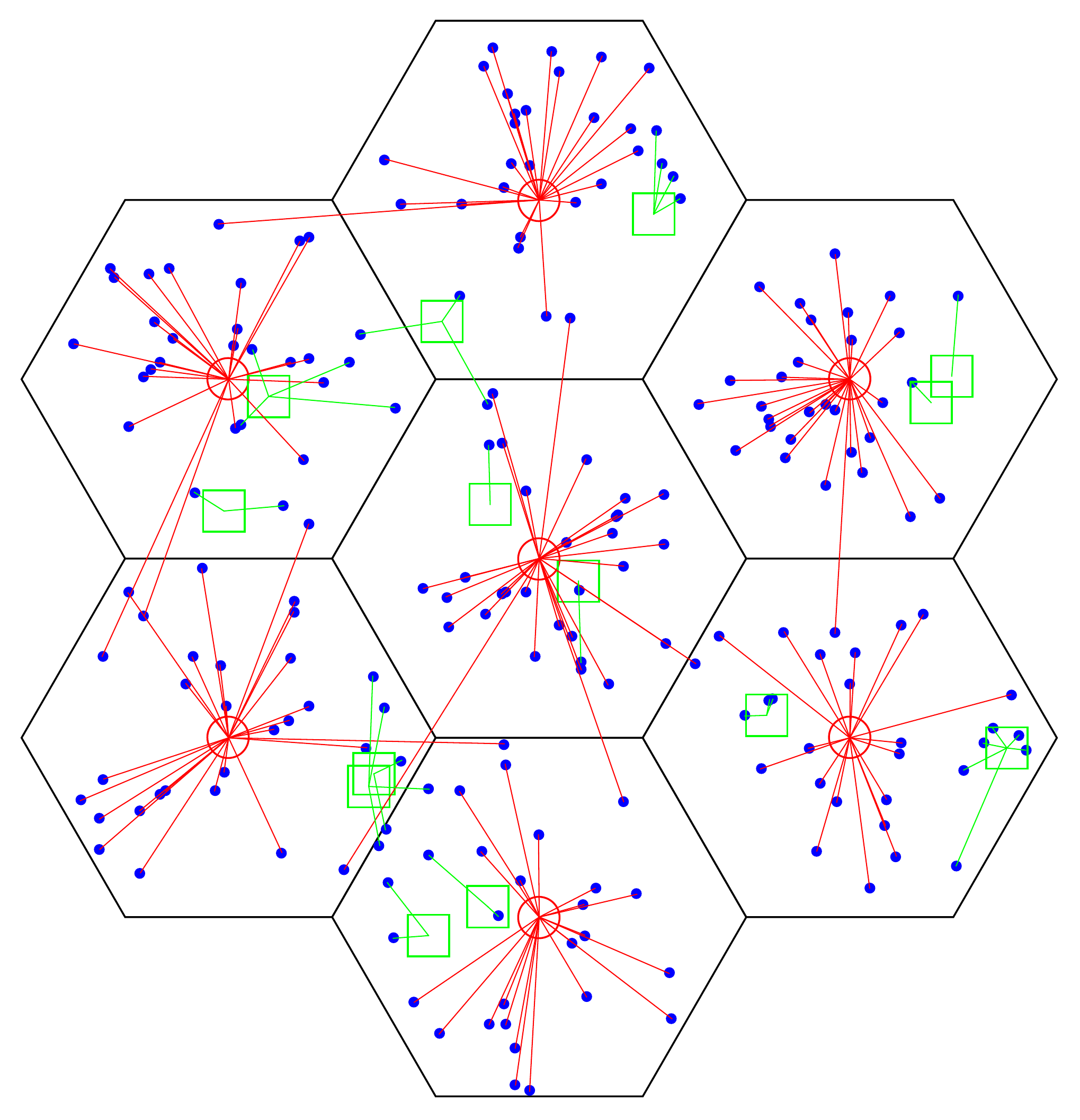}}~~~
\subfigure[Load with respect to user demand]
{\includegraphics[width=0.3\linewidth]{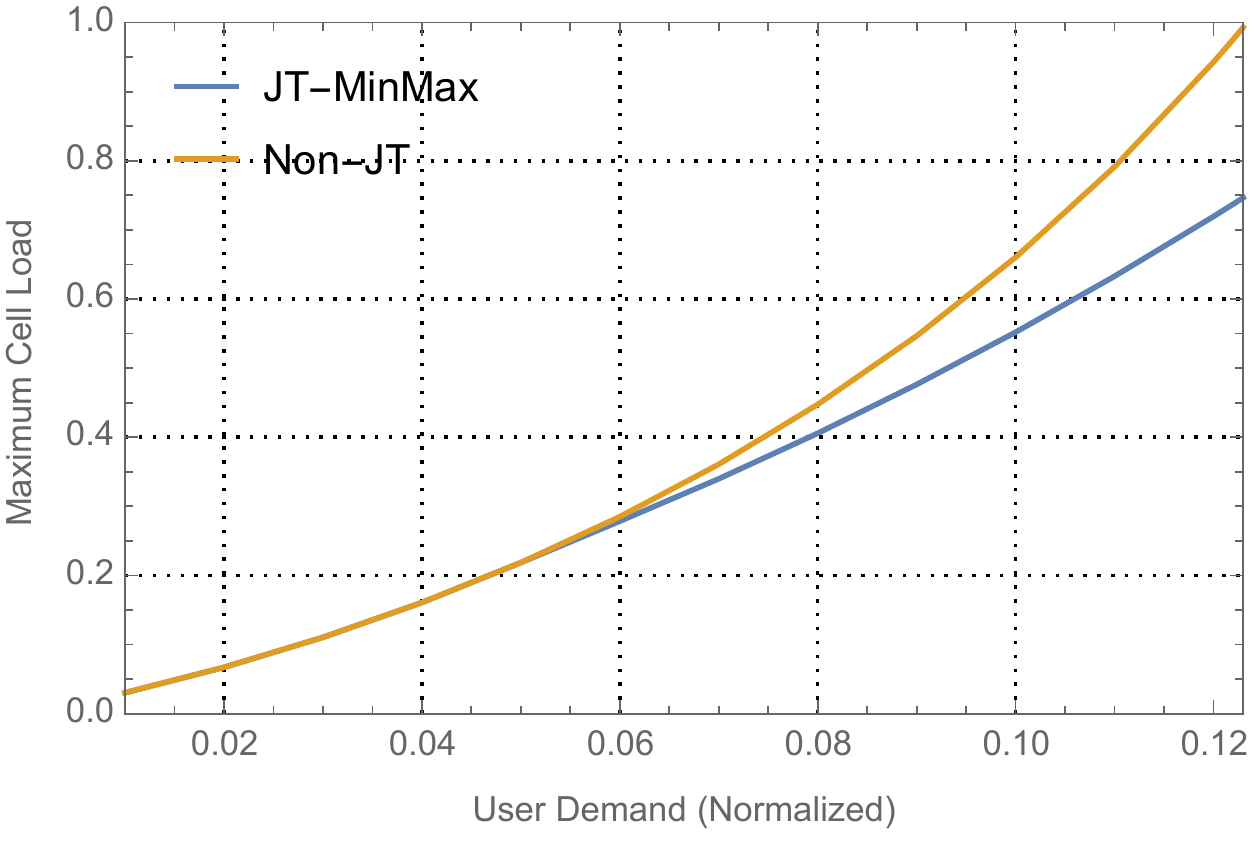}}~~~
\subfigure[Load for JT-MinMax and Non-JT in each cell.]
{\includegraphics[width=0.32\linewidth]{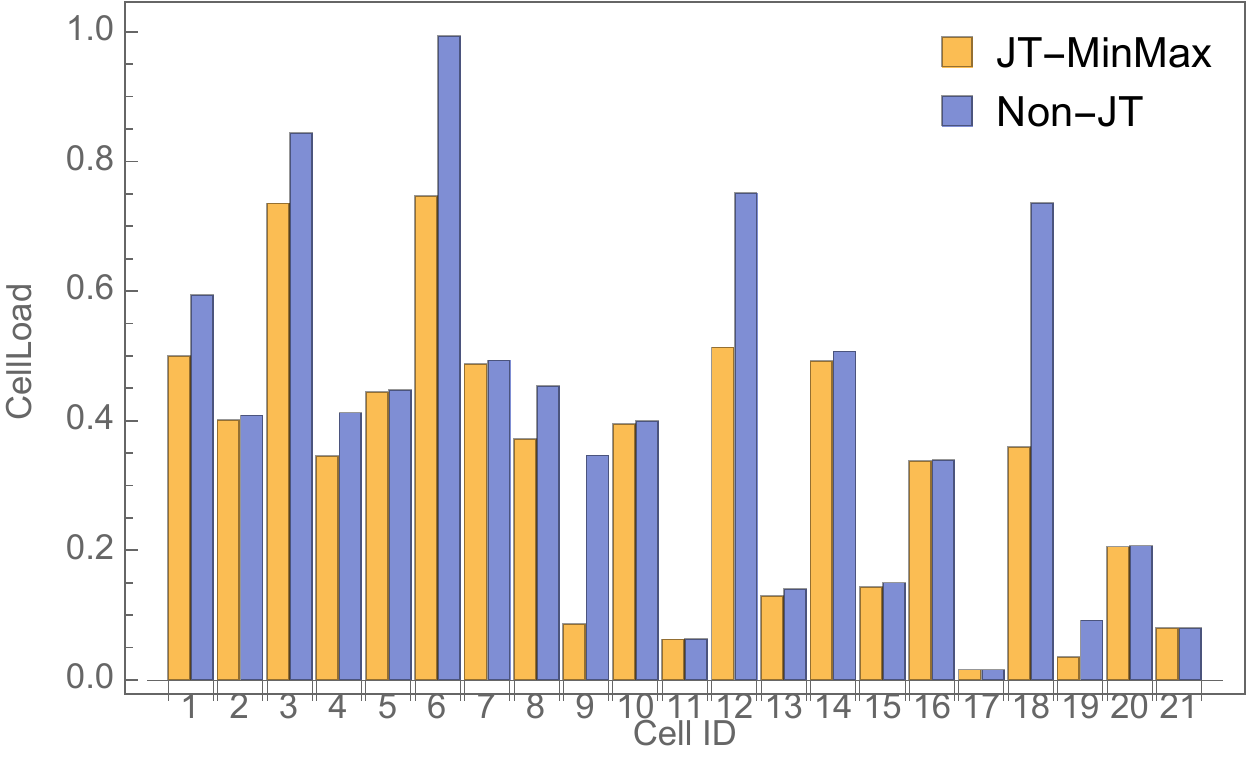}}
\caption{The sub-figure (a) illustrates the HetNet layout, where MC, SC and UE are denoted by $\bigcirc$, $\Box$, and $\bm{\cdot}$ respectively. Sub-figures (b) and (c) shows the numerical results.}
\vskip -15pt
\label{fig:network}
\end{figure*}

The network layout, as deployed in \cite{Anonymous:SLcg4Bln}, is illustrated in \figurename~\ref{fig:network}~(a). There are seven hexagonal cell regions in total, of which the center is deployed with one MC, indexed by numbers 1--7. Two SCs are randomly placed in each hexagon, indexed by numbers 8--21. Thirty UEs are randomly and uniformly distributed in each hexagonal region. The HetNet operates at 2 GHz. Each resource block follows the LTE standard of 180 kHz bandwidth and the bandwidth for each cell is 4.5 MHz. The transmit power per resource block for MCs and SCs are set to 200 mW and 50 mW, respectively. The noise power spectral density is set to -174 dBm/Hz. The path loss follows COST-231-HATA model and the shadowing coefficients are generated by the log-normal distribution with 8 dB standard deviation. The network is initialized to that each UE is connected to the cell (MC or SC) with the best received signal power. In JT-MinMax, $\gamma$ and $\tau$ are set to 5 and 20, respectively. 

In \figurename~\ref{fig:network}~(b), we compare the maximum cell load in the initial non-JT case with JT-MinMax, with respect to user demands. In the initial non-JT case, we let each UE select its serving cell by the best received signal power. By JT-MinMax, the maximum cell load is reduced by 17.82\% on average. For the maximum achievable user demand, the maximum cell load is improved by 24.81\%. We can see from \figurename~\ref{fig:network}~(b) that the improvement is much more significant in higher user demand. When the user demand is very low, there is visually no difference between the two schemes. \figurename~\ref{fig:network}~(c) shows the load of each cell for the maximum achievable user demand. We can see from the result that, by JT-MinMax, the load is reduced for both MCs and SCs. On average, the difference between maximum and minimum cell load is reduced by 25.21\% via JT-MinMax. Specifically, the apparently high load of some cells, e.g., cells $1$, $3$, $6$, $12$ and $18$, are reduced.




\section{Conclusion}
\label{sec:conclusion}
This paper proposed a generalized version of the load coupling model, taking into account JT. The load coupling equations are given for both cells and UEs. Then we show two observations in terms of SIF property of these load coupling equations. Under the proposed model, the load balancing problem is studied for both a special two-cell case and the general scenario in HetNet. For the special two-cell case, we show that the greedy algorithm achieves the global optimality. For the general case, a sufficient condition for reducing the cell load is given, which is utilized as a criterion for adding the JT downlinks in the proposed algorithm JT-MinMax. For the HetNet scenario, JT-MinMax leads to better performance in maximum cell load than Non-JT solution.

\section*{Acknowledgements}
\vskip -4pt
This work has been supported by the Swedish Research Council and the Link\"{o}ping-Lund Excellence Center in Information Technology (ELLIIT), Sweden, and the European Union Marie Curie project MESH-WISE (FP7-PEOPLE-2012-IAPP: 324515). The work of the second author has been partially supported by the China Scholarship Council (CSC). The work of D. Yuan has been carried out within European FP7 Marie Curie IOF project 329313.


\end{document}